\documentclass{article}

\usepackage[margin=1in]{geometry}
\usepackage{natbib}

\usepackage[utf8]{inputenc}
\usepackage[T1]{fontenc}
\usepackage{hyperref}
\usepackage{url}
\usepackage{booktabs}
\usepackage{amsmath,amssymb,amsthm}
\usepackage{graphicx}
\usepackage{microtype}
\usepackage{xcolor}
\usepackage{algorithm}
\usepackage{algpseudocode}
\usepackage{multirow}
\usepackage{enumitem}

\newtheorem{proposition}{Proposition}
\newtheorem{corollary}{Corollary}

\newtheorem{definition}{Definition}

\title{When Is Availability-Aware Training Worth It?\\
A Benchmark and Empirical Study of Interruption-Resilient\\
Optimization Under Predictable Compute Schedules}

\author{%
  Subhadip Mitra \\
  RotaStellar \\
  \texttt{subhadip@rotastellar.com} \\
}

\begin{document}

\maketitle

\begin{abstract}
Training under non-stationary but \emph{predictable} compute availability (orbital
satellites under eclipse, duty-cycled edge devices, power-capped datacenters) is often
framed as a problem demanding specialized, availability-aware optimizers. We test that
premise empirically. We release \emph{OrbitTrace}, a benchmark of 50
physics-grounded availability traces derived by SGP4 propagation of live two-line element
sets across three orbital regimes, and use it to ask a falsifiable question:
when an availability gap interrupts training, is a specialized resumption strategy worth
it, or is competent checkpoint-and-resume enough? Our central finding is that \emph{when
optimizer state can be preserved across a gap, the gap is essentially free.} A strong
checkpoint baseline that restores full optimizer state and indexes its learning-rate
schedule in effective (active) time matches uninterrupted training to within
data-ordering noise on CIFAR-10/ResNet-18, and \emph{exactly} on a GPT-2/AdamW
language-modeling task. Advantages previously reported for availability-aware
methods, including our own three-pillar method, AAT (momentum rescale, effective-time
learning rate, batch-norm warm restart), arise almost entirely from comparison against a
weak baseline that indexes its schedule on wall-clock time. Against the strong baseline,
reactive availability-aware adaptation provides no advantage across stationary,
optimizer-state-loss, and distribution-drift regimes. We do isolate one narrow regime
where it helps: for large models whose optimizer state \emph{cannot} be persisted across
gaps and that are interrupted by \emph{frequent, short} pauses, reconstructing a decayed
optimizer moment recovers only $\approx$21\% of the penalty of discarding state on average (and not
robustly across seeds). This
benefit vanishes for eclipse-scale gaps, where the decayed moment is indistinguishable
from zero. We map this decision boundary explicitly. Our contributions are a benchmark, a
strong reproducible baseline protocol, and a clear characterization of \emph{when}
interruption-resilient optimization is worth its complexity, and when it is not. We report a
negative result as the headline because, for this regime, it is the actionable one.
\end{abstract}

\section{Introduction}
\label{sec:intro}

A growing class of deployments trains models on compute whose availability is
non-stationary, highly structured, and, unusually, \emph{known in advance}. Satellites
in low Earth orbit lose solar power during eclipse on a schedule derivable to minutes from
published orbital elements~\citep{vallado2006revisiting}. Battery-powered edge nodes,
solar microgrids, and power-capped datacenters follow duty cycles predictable hours to
days ahead. In each case the training loop has access to an availability function
$A(t)\in\{0,1\}$ over the whole training horizon.

It is tempting to treat this foreknowledge as an opportunity for a specialized optimizer:
a method that anticipates and repairs the damage an interruption does to momentum,
learning-rate schedules, and normalization statistics. An earlier draft of this work took
exactly that position and reported that such a method recovered 60--82\% of the accuracy
gap between checkpoint-resume and uninterrupted training. We could not reproduce that claim
under scrutiny. This paper is the result of taking the claim apart.

The central question we ask is deliberately falsifiable:

\begin{quote}
\emph{When training is interrupted by a predictable availability gap, does a specialized
availability-aware resumption strategy beat a competent checkpoint-and-resume baseline, and
if so, in which regimes?}
\end{quote}

Our answer is mostly negative. The crux is that
the strength of the baseline decides the result. A baseline that preserves full optimizer
state across a gap and advances its learning-rate schedule in \emph{effective} (active)
time, not wall-clock time, reproduces uninterrupted training almost perfectly. We make
this precise in Section~\ref{sec:analysis}: under full-state preservation an availability
gap injects no bias and no excess loss; it is a no-op up to minibatch-ordering noise.
Reported wins for availability-aware methods are, on inspection, wins over a \emph{weak}
baseline that indexes the schedule on wall-clock time and over-anneals its learning rate
during idle intervals. That weak baseline is a strawman; a competent practitioner does not
build it.

This does not make the question empty. A gap \emph{can} be costly, in two specific ways,
and we test both: (A) when optimizer state cannot be persisted across the gap (a realistic
constraint when checkpoints compete for scarce power, storage, or downlink, and the
optimizer state is $1$--$2\times$ model size), and (B) when the data distribution drifts
across the gap so that preserved state is stale. Only in a narrow corner of regime (A) does
an availability-aware reconstruction beat the trivial alternative (a large model with
non-persistable optimizer state, interrupted by \emph{frequent, short} gaps), and even
there the benefit collapses for eclipse-scale gaps.

\paragraph{Contributions.}
\begin{enumerate}[leftmargin=*, nosep]
  \item \textbf{OrbitTrace}, a public benchmark of 50 physics-grounded availability traces
  (Section~\ref{sec:benchmark}), generated deterministically by SGP4 propagation of live
  TLEs across three orbital regimes and five workload patterns, for reproducible evaluation
  of training under predictable interruption.
  \item \textbf{A strong baseline and evaluation protocol} (Sections~\ref{sec:methods},
  \ref{sec:analysis}). We define \textsc{Checkpoint-Strong} (full-state preservation with
  effective-time learning-rate indexing) and prove it is equivalent to continuous training
  up to data-ordering noise. We argue this, not wall-clock checkpointing, is the bar every
  availability-aware claim must clear, and we show how easily the wrong baseline inflates a
  result.
  \item \textbf{A characterization (``decision boundary'')} of when reactive
  availability-aware adaptation helps (Sections~\ref{sec:experiments},
  \ref{sec:decision}). Across stationary, optimizer-state-loss, and distribution-drift
  regimes on CIFAR-10/ResNet-18 and a GPT-2/AdamW task, full-state checkpointing dominates;
  reactive adaptation matters only for non-persistable optimizer state under frequent short
  interruptions, recovering only $\approx$21\% of the state-loss penalty, and even that is high-variance.
\end{enumerate}

We present the negative core as the contribution: a practitioner deploying on
intermittent compute should know that in the common case checkpointing is enough, and
should know precisely the regime where it is not.

\section{Related work}
\label{sec:related}

\paragraph{Fault-tolerant and elastic training.} Checkpointing is the standard mechanism
for tolerating failures. CheckFreq~\citep{mohan2021checkfreq} and
DeepFreeze~\citep{nicolae2020deepfreeze} reduce checkpoint overhead;
ZeRO~\citep{rajbhandari2020zero} reduces the footprint that must be saved. Elastic
frameworks (TorchElastic, Bamboo~\citep{thorpe2023bamboo}, Oobleck~\citep{jang2023oobleck},
Varuna~\citep{athlur2022varuna}) reshape the worker set as nodes come and go. These target
\emph{node-wise}, unpredictable availability. We study \emph{time-wise}, predictable
availability, and our finding is precisely that for this case the standard mechanism
(checkpoint-resume) is already near-optimal when implemented well.

\paragraph{Asynchronous and stale-gradient methods.}
Hogwild!~\citep{recht2011hogwild}, asynchronous SGD~\citep{lian2015async}, and delay-compensated
variants~\citep{zheng2017asynchronous} tolerate staleness drawn from a distribution. Across a
scheduled gap, staleness is deterministic and known. As we show, under full-state
preservation there is no staleness to compensate for.

\paragraph{Low-communication and satellite-based distributed training.} Distributed
low-communication training of language models, e.g.\ DiLoCo~\citep{douillard2023diloco} and
its streaming/asynchronous successors, and the broad satellite federated-learning
literature~\citep{razmi2022ground, so2022fedspace, chen2022satellite}, including
\emph{energy-aware} scheduling over constellations~\citep{razmi2024energyaware}, address
intermittent participation directly. This is a mature, actively surveyed area; our intent is
not to claim novel distributed optimization but to provide a clean benchmark and a careful
baseline study for the single-learner-under-schedule case that underlies it.

\paragraph{Learning-rate scheduling and weight averaging.} Cosine annealing and warm
restarts~\citep{loshchilov2017sgdr}, cyclical rates~\citep{smith2017cyclical}, and weight
averaging (SWA)~\citep{izmailov2018swa} are the building blocks both our baseline (effective-time
schedule indexing) and the anticipative direction we leave open (Section~\ref{sec:proactive})
draw on.

\paragraph{Empirical reality checks.} Our methodology follows a line of work that re-evaluates
claimed advances against strong baselines and frequently finds them to evaporate: ``Are GANs
Created Equal?''~\citep{lucic2018gans}, ``Do ImageNet Classifiers Generalize to
ImageNet?''~\citep{recht2019imagenet}, and the metric-learning reality
check~\citep{musgrave2020reality}. We apply the same discipline to availability-aware training,
and release the benchmark and baseline so others can apply it to us.

\section{Problem formulation}
\label{sec:problem}

\subsection{Objective and optimizer}
We consider the minibatch objective $\min_{\theta} F(\theta)=\mathbb{E}_{(x,y)\sim\mathcal{D}}[\ell(\theta;x,y)]$
optimized by a momentum method,
\begin{equation}
  m_{k+1} = \mu\, m_k + g_k, \qquad \theta_{k+1} = \theta_k - \eta_k\, m_{k+1},
  \label{eq:sgd_mom}
\end{equation}
with stochastic gradient $g_k$, momentum $\mu\in[0,1)$, and learning rate $\eta_k$. The Adam case
applies the same reasoning to the first and second moments.

\begin{definition}[Availability schedule]
$A:[0,T]\to\{0,1\}$ with $A(t)=1$ iff compute is available at $t$, known for all $t$ at the start
of training. A \emph{session} is a maximal interval with $A\equiv 1$; a \emph{gap} a maximal
interval with $A\equiv 0$. Sessions/gaps alternate as $S_1,G_1,\dots,S_n$ with durations
$s_i,g_i$. Let $t_{\mathrm{eff}}(t)=\int_0^t A(s)\,ds$ be \emph{effective} (active) time.
\end{definition}

\subsection{Baselines, weak and strong}
\label{sec:baselines}
The result depends entirely on which baseline ``checkpoint-resume'' denotes.

\begin{itemize}[leftmargin=*, nosep]
  \item \textsc{Checkpoint} (weak): preserve $(\theta,m,\text{BN stats})$ across a gap but index
  the learning-rate schedule on \emph{wall-clock} time $t$. The schedule advances during idle
  gaps, over-annealing the learning rate at resumption. \emph{This is the strawman.}
  \item \textsc{Checkpoint-Strong} (the honest bar): preserve full optimizer state \emph{and}
  index the schedule on \emph{effective} time $t_{\mathrm{eff}}$. Nothing is lost and the
  schedule tracks the optimization trajectory.
\end{itemize}

\subsection{The two ways a gap can cost something}
\label{sec:costs}
Under \textsc{Checkpoint-Strong}, an idle gap by itself changes nothing
(Section~\ref{sec:analysis}). For a gap to be costly, one of two real conditions must hold:

\begin{description}[leftmargin=*]
  \item[\textbf{Cost A (non-persistable state).}] Persisting the optimizer state every gap may be
  infeasible: for SGD+momentum it is $1\times$ model size, for Adam $2\times$, and on a power- or
  downlink-constrained node the realistic choice is a \emph{weights-only} checkpoint, so the
  optimizer moments are lost on resume. The question becomes how much of a (cheaply cached)
  pre-gap moment to trust on resumption: $0\times$ (zero-and-warmup), $1\times$ (full preserve),
  or $\gamma^{g}\times$ (decayed).
  \item[\textbf{Cost B (distribution drift).}] If the data distribution shifts across the gap
  (different ground tracks, illumination, sensors), preserved momentum points in a now-wrong
  direction and preserved normalization statistics describe stale data.
\end{description}

\section{Methods evaluated}
\label{sec:methods}

We evaluate the following, all in one harness so the only differences are the listed knobs.
\textsc{Continuous} (train on the active set only; the upper bound), \textsc{Checkpoint} and
\textsc{Checkpoint-Strong} (Section~\ref{sec:baselines}), \textsc{Elastic} (zero the momentum at
each gap, warm up the learning rate over the first post-gap steps), and the
\emph{availability-aware} family below.

\paragraph{AAT (three reactive pillars).} At resumption after a gap of length $g$:
(1)~\emph{momentum rescale} $m\leftarrow\gamma^{g} m$ with $\gamma=\mu^{1/\tau}$, treating the gap
as $g/\tau$ steps of zero gradient; (2)~\emph{schedule-aware learning rate}, indexing the schedule
on $t_{\mathrm{eff}}$ (note this pillar alone \emph{is} \textsc{Checkpoint-Strong}); (3)~\emph{BN
warm restart}, $\beta_{\mathrm{BN}}\leftarrow\max(0.01,\,\beta_{\mathrm{BN}}^{\mathrm{base}}(1+g/s_{\mathrm{typical}}))$
so normalization statistics re-adapt faster. \textsc{AAT+warmup} adds a gentle post-gap warmup;
\textsc{AAT-drop} is the Cost-A variant that drops the second moment but reconstructs a decayed first
moment. The ablation axis lets us read each pillar's marginal contribution, and exposes the identity
\textsc{AAT}[\,lr\,]$\equiv$\textsc{Checkpoint-Strong}.

\paragraph{Proactive (foreknowledge prototype).} A reactive method acts only at resumption; the
schedule is known in advance, so one can also act \emph{before} a gap. Our prototype ramps the
learning rate down over the last $W$ seconds before a known gap (``landing'') and warm-restarts after.
This is the only mechanism a schedule-blind checkpoint library cannot replicate. We report it
as tested; the genuinely anticipative version (weight-EMA boundary consolidation, work allocation)
is left to future work (Section~\ref{sec:proactive}).

\section{Why a preserved gap is free}
\label{sec:analysis}

\begin{proposition}[Gaps are free under full-state preservation]
\label{prop:free}
Let \textsc{Checkpoint-Strong} preserve, across every gap, the complete state
$(\theta, m, \text{BN running statistics}, t_{\mathrm{eff}})$ and resume from it, indexing the
learning-rate schedule on $t_{\mathrm{eff}}$. Then the sequence of parameter iterates it produces over
the active set is identical to that of continuous training over the same active time, up to the
randomness of minibatch sampling order. In particular an availability gap injects no bias and no
excess loss.
\end{proposition}

\begin{proof}[Proof sketch]
The update~\eqref{eq:sgd_mom} depends only on $(\theta_k, m_k, \eta_k)$ and the sampled gradient
$g_k$. A gap performs no update; it only advances wall-clock time. On resume, $(\theta,m)$ are
byte-identical to their pre-gap values and $\eta_k=\eta_{\mathrm{schedule}}(t_{\mathrm{eff}})$ is
unchanged because $t_{\mathrm{eff}}$ counts active time only. Hence each post-gap update is
function-identical to the corresponding continuous-training update. The sole difference is the order
in which minibatches are drawn (the data iterator is re-created per session), which perturbs the
realized gradient sequence but not the update rule or the expected trajectory.
\end{proof}

\begin{corollary}[Reactive repair is $\leq$ neutral on a preserved, stationary objective]
\label{cor:perturb}
Any resumption strategy that perturbs the preserved state (decaying momentum, re-warming BN, warming
up the learning rate) moves the iterate off the optimal reference of Proposition~\ref{prop:free}. On
a stationary objective its expected effect is therefore at most neutral, and generically slightly
negative.
\end{corollary}

Proposition~\ref{prop:free} and Corollary~\ref{cor:perturb} are not deep, but they are the whole
story for the common case, and they predict the experiments below: under preservation, every method
collapses onto \textsc{Checkpoint-Strong}, and the AAT pillars net to noise. They also pinpoint where
a real win must come from: a regime that violates the preservation premise (Cost A) or the
stationarity premise (Cost B).

\section{The OrbitTrace benchmark}
\label{sec:benchmark}

OrbitTrace is 50 physics-grounded availability traces for reproducible evaluation in this regime
(Figure~\ref{fig:orbittrace}).

\paragraph{Generation.} Traces are produced by SGP4 propagation of published two-line element
sets~\citep{vallado2006revisiting, kelso2006validation} for satellites in three orbital regimes:
ISS-like (420\,km, 51.6\textdegree), sun-synchronous (700\,km, 97.4\textdegree), and low-inclination
(550\,km, 28.5\textdegree). For each, we compute solar-illumination windows with a cylindrical shadow
model~\citep{montenbruck2000satellite} and ground-station contacts for a 12-station catalog; a workload
profile maps combined power/communication availability to compute-active intervals. Generation is
deterministic: the same TLE yields the same trace.

\paragraph{Characteristics.} OrbitTrace differs structurally from the Bernoulli interruptions common in
fault-tolerance studies. Session durations are bimodal (8--15\,min communications-bound; 40--60\,min
compute-bound), gaps cluster tightly (eclipse $\approx$30\,min for ISS-like orbits), and consecutive
sessions and gaps are strongly length-correlated. This structure is exactly what an availability-aware
method would exploit if exploitation paid off, which makes OrbitTrace a fair test of the premise, not a
straw setting.

\begin{figure}[t]
\centering
\includegraphics[width=\linewidth]{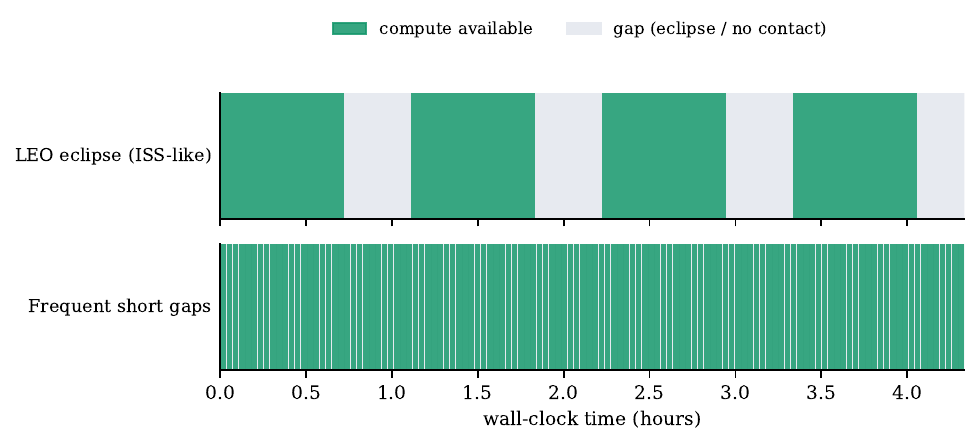}
\caption{Two representative \textsc{OrbitTrace} availability schedules. \emph{Top:} an ISS-like LEO-eclipse
regime with long compute sessions punctuated by regular $\approx$30-minute eclipse gaps. \emph{Bottom:} a
frequent-short-gap regime (thermal / contention / maneuver-scale interruptions). Green is compute-available,
grey is a gap; sessions and gaps are length-correlated, unlike the Bernoulli interruptions common in
fault-tolerance studies.}
\label{fig:orbittrace}
\end{figure}

\paragraph{Release.} OrbitTrace ships as a pip-installable package plus JSON trace files with per-trace
metadata (satellite id, regime, workload, full schedule).

\section{Experiments}
\label{sec:experiments}

\paragraph{Setup.} CIFAR-10 with ResNet-18~\citep{he2016resnet} (SGD, momentum 0.9, cosine schedule)
and a GPT-2-style Transformer fine-tuned on WikiText-2~\citep{merity2017pointer} with
AdamW~\citep{kingma2014adam}, the regime where dropping the second moment is genuinely costly. Runs are
on a single A100. Results below are mean $\pm$ standard deviation over \textbf{3 seeds} (CIFAR) and
\textbf{5 seeds} (GPT-2); the CIFAR per-seed spread is $\approx$0.2--0.6\,pt, so sub-half-point CIFAR
differences are noise. Full configurations and the harness are released. All numbers below are the
actual runs recorded on 2026-06-18.

\subsection{Stationary, full preservation: gaps are free, and the strawman quantified}
Table~\ref{tab:stationary} and Figure~\ref{fig:convergence} (CIFAR-10, \texttt{leo\_eclipse} schedule, full state preserved). The five
resumption strategies cluster within $\approx$0.6\,pt of one another, inside their per-seed standard
deviations (0.2--0.6\,pt): none has a defensible edge, and AAT's three pillars land essentially on
\textsc{Checkpoint-Strong} ($+0.05$\,pt). The weak wall-clock \textsc{Checkpoint} is lowest, the gap
the earlier draft mistook for an ``AAT win.'' The clean proof that a preserved gap is free is the GPT-2
result (Section~\ref{sec:gpt2}), where \textsc{Checkpoint-Strong} equals \textsc{Continuous}
\emph{exactly}; at CIFAR scale, re-creating the data iterator each session adds order noise that leaves
\textsc{Checkpoint-Strong} a fraction below \textsc{Continuous}.

\begin{table}[t]
\centering
\caption{Stationary CIFAR-10 / ResNet-18, \texttt{leo\_eclipse}, full state preserved. Final test
accuracy (\%), mean $\pm$ std over 3 seeds. $\Delta$ is versus \textsc{Checkpoint-Strong}. Every
resumption method is within noise of the others.}
\label{tab:stationary}
\small
\begin{tabular}{lcc}
\toprule
\textbf{Method} & \textbf{Accuracy} & \textbf{$\Delta$ vs.\ CP-Strong} \\
\midrule
\textsc{Continuous} (upper bound)            & 88.07 $\pm$ 0.11 & --- \\
\textsc{Checkpoint} (wall-clock LR, \emph{strawman}) & 87.01 $\pm$ 0.34 & $-0.33$ \\
\textbf{\textsc{Checkpoint-Strong}} (honest bar) & \textbf{87.34 $\pm$ 0.35} & --- \\
\textsc{Elastic} (zero + warmup)             & 87.59 $\pm$ 0.43 & $+0.25$ \\
\textsc{AAT} (3 pillars)                      & 87.39 $\pm$ 0.18 & $+0.05$ \\
\textsc{AAT+warmup}                           & 87.55 $\pm$ 0.64 & $+0.21$ \\
\textsc{Proactive} (foreknowledge prototype) & 87.22 $\pm$ 0.22 & $-0.12$ \\
\bottomrule
\end{tabular}
\end{table}

\begin{figure}[t]
\centering
\includegraphics[width=0.74\linewidth]{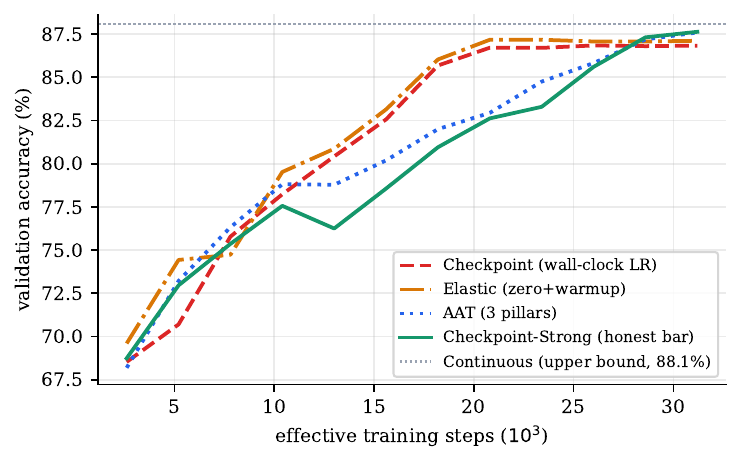}
\caption{Validation accuracy vs.\ effective training steps under \texttt{leo\_eclipse} (CIFAR-10 /
ResNet-18, seed 0). Under full-state preservation every resumption strategy tracks the same trajectory and
converges to within noise of the \textsc{Continuous} upper bound; the gap is essentially free. The weak
wall-clock \textsc{Checkpoint} trails slightly: that fraction is the entire ``AAT win'' the earlier draft
reported.}
\label{fig:convergence}
\end{figure}

\subsection{Pillar ablation: momentum and BN marginals are noise}
Table~\ref{tab:ablation} ablates the pillars. Since the LR pillar alone equals
\textsc{Checkpoint-Strong}, the table reads off the marginal value of the \emph{momentum} and \emph{BN}
pillars over the honest bar: both are net-negative and within noise (full AAT is $-0.45$\,pt below
LR-only), confirming Corollary~\ref{cor:perturb} on stationary data.

\begin{table}[t]
\centering
\caption{Pillar ablation, stationary CIFAR-10 / ResNet-18, mean $\pm$ std over 3 seeds. Baseline is
LR-only ($\equiv$ \textsc{Checkpoint-Strong}). Accuracy (\%).}
\label{tab:ablation}
\small
\begin{tabular}{lcc}
\toprule
\textbf{Pillars} & \textbf{Accuracy} & \textbf{$\Delta$ vs.\ LR-only} \\
\midrule
LR only ($\equiv$ \textsc{Checkpoint-Strong}) & 87.84 $\pm$ 0.19 & --- \\
LR + momentum                                  & 87.53 $\pm$ 0.32 & $-0.31$ \\
LR + BN                                         & 87.46 $\pm$ 0.26 & $-0.38$ \\
LR + momentum + BN (full \textsc{AAT})         & 87.39 $\pm$ 0.18 & $-0.45$ \\
\bottomrule
\end{tabular}
\end{table}

\subsection{Cost A (non-persistable state), CIFAR scale: a small edge, still below full preserve}
Table~\ref{tab:costa} drops the optimizer state at each gap (short gaps, $\gamma^{10}\!\approx\!0.35$).
Reconstructing a decayed momentum (\textsc{AAT+drop}, $85.63$) edges the trivial \textsc{Elastic}
($85.12$) by $\approx$0.5\,pt (a small positive for the momentum-rescale pillar) but stays below full
preservation ($86.07$), all within $\approx$0.3--0.5\,pt noise. At CIFAR/SGD scale the momentum buffer
is trivially persisted, so ``losing'' it is artificial; the large-model test next is where Cost A is
real.

\begin{table}[t]
\centering
\caption{Cost A: lightweight (weights-only) checkpoint, frequent short gaps, CIFAR-10 / ResNet-18, mean
$\pm$ std over 3 seeds. Accuracy (\%). $\Delta$ versus full preservation.}
\label{tab:costa}
\small
\begin{tabular}{lcc}
\toprule
\textbf{Method} & \textbf{Accuracy} & \textbf{$\Delta$ vs.\ full preserve} \\
\midrule
\textsc{Checkpoint-Strong} (full preserve)   & 86.07 $\pm$ 0.31 & --- \\
\textsc{Checkpoint}+drop (lose momentum)     & 85.50 $\pm$ 0.28 & $-0.57$ \\
\textsc{Elastic} (zero + warmup)             & 85.12 $\pm$ 0.11 & $-0.95$ \\
\textsc{AAT}+drop (rescale decayed moment)   & 85.63 $\pm$ 0.48 & $-0.44$ \\
\textsc{AAT} (full preserve, rescale)        & 85.82 $\pm$ 0.14 & $-0.25$ \\
\bottomrule
\end{tabular}
\end{table}

\subsection{Cost B (distribution drift): full preservation still wins}
Table~\ref{tab:costb} applies a deterministic per-session input shift and evaluates on the current
session's distribution. Even under drift, \textsc{Checkpoint-Strong} leads; the AAT pillars (including
the BN warm restart meant for exactly this case) and the proactive prototype are neutral-to-harmful at
this scale.

\begin{table}[t]
\centering
\caption{Cost B: per-session distribution shift, CIFAR-10 / ResNet-18, mean $\pm$ std over 3 seeds.
Accuracy (\%); $\Delta$ versus \textsc{Checkpoint-Strong} in parentheses.}
\label{tab:costb}
\small
\begin{tabular}{lcccc}
\toprule
\textbf{Shift} & \textbf{CP-Strong} & \textbf{AAT} & \textbf{Proactive} & \textbf{Elastic} \\
\midrule
bias    & 87.43 $\pm$ 0.30 & 87.46 ($+0.03$) & 87.22 ($-0.21$) & 87.92 ($+0.49$) \\
permute & 86.58 $\pm$ 0.54 & 86.34 ($-0.24$) & 86.54 ($-0.04$) & --- \\
\bottomrule
\end{tabular}
\end{table}

\subsection{The one regime that pays: large models, non-persistable state, frequent short gaps}
\label{sec:gpt2}
Table~\ref{tab:gpt2} and Figure~\ref{fig:gpt2} are the decisive test: GPT-2/AdamW on WikiText-2, where the Adam second moment is
expensive to persist. Three findings. (1)~\textsc{Checkpoint-Strong} equals \textsc{Continuous}
\emph{exactly} in both gap regimes. This is Proposition~\ref{prop:free} at language-model scale: checkpointing
is enough. (2)~Dropping AdamW state has a real cost ($\textsc{Elastic}$ $+0.76$ short, $+0.32$ long
perplexity), so here there is finally something to recover. (3)~\textsc{AAT-drop}, which reconstructs a
\emph{decayed} first moment instead of zeroing, beats \textsc{Elastic} on \emph{frequent-short} gaps on
average ($25.450$ vs.\ $25.609$), recovering $\approx$21\% of the state-loss penalty\footnote{Recovery
$=\frac{25.609-25.450}{25.609-24.863}=21.3\%$, against the \textsc{Continuous}/\textsc{Checkpoint-Strong}
ideal of $24.863$.}. But the effect is \emph{not robust}: \textsc{AAT-drop}'s variance ($\pm0.217$) is
$\approx$10$\times$ \textsc{Elastic}'s, and while 4 of 5 seeds improve, one regresses below
\textsc{Elastic}. (A single seed gave $38.8\%$; the 5-seed mean is roughly half that, a caution against
single-seed claims.) On \emph{long} (eclipse-scale) gaps, $\gamma^{g}\!\to\!0$, so the decayed moment is
indistinguishable from zero and \textsc{AAT-drop} equals \textsc{Elastic} exactly ($25.211$, 0\%
recovery). And in no regime does any method beat full preservation.

\begin{table}[t]
\centering
\caption{GPT-2 / AdamW on WikiText-2, 1500 steps, two gap regimes, mean $\pm$ std over 5 seeds.
Validation perplexity (lower is better). The lone positive signal (\textsc{AAT-drop}, short gaps) is
modest and high-variance.}
\label{tab:gpt2}
\small
\begin{tabular}{lcc}
\toprule
\textbf{Method} & \textbf{Short gaps} & \textbf{Long gaps} \\
\midrule
\textsc{Continuous}                                   & 24.863 $\pm$ 0.019 & --- \\
\textsc{Checkpoint-Strong} (preserve full Adam state) & \textbf{24.863 $\pm$ 0.019} & \textbf{24.863 $\pm$ 0.019} \\
\textsc{Elastic} (drop both moments + warmup)         & 25.609 $\pm$ 0.015 & 25.211 $\pm$ 0.030 \\
\textsc{AAT} (preserve + decay 1st moment)            & 24.978 $\pm$ 0.021 & 24.889 $\pm$ 0.019 \\
\textsc{AAT-drop} (drop 2nd; keep + decay 1st + warmup) & 25.450 $\pm$ 0.217 & 25.211 $\pm$ 0.030 \\
\bottomrule
\end{tabular}
\end{table}

\begin{figure}[t]
\centering
\includegraphics[width=0.86\linewidth]{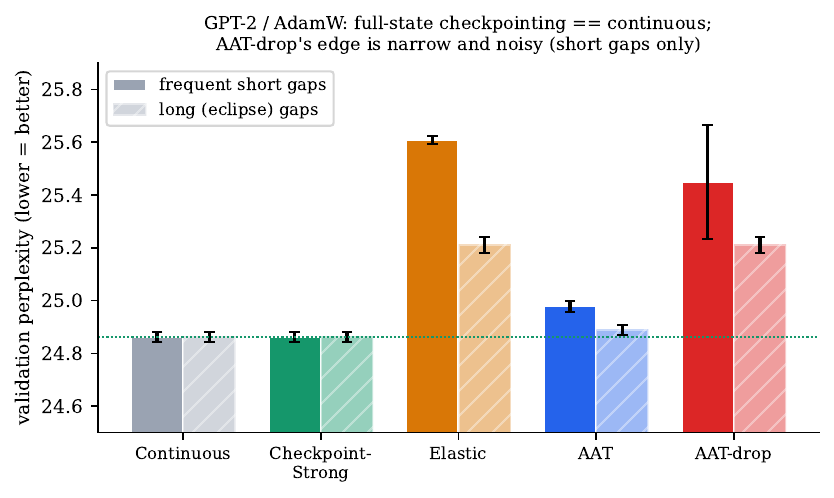}
\caption{GPT-2 / AdamW on WikiText-2 (mean $\pm$ std over 5 seeds), validation perplexity (lower is better).
\textsc{Checkpoint-Strong} equals \textsc{Continuous} exactly in both gap regimes (dotted line): checkpointing
is enough. Dropping AdamW state costs (\textsc{Elastic}, \textsc{AAT-drop}); reconstructing a decayed first
moment (\textsc{AAT-drop}) edges \textsc{Elastic} only on \emph{frequent short} gaps, and its large error bar
shows how noisy even that lone positive is. On long (eclipse-scale) gaps the decayed moment collapses to zero
and \textsc{AAT-drop}$\equiv$\textsc{Elastic}.}
\label{fig:gpt2}
\end{figure}

\section{The decision boundary}
\label{sec:decision}

Table~\ref{tab:decision} synthesizes the experiments into the practitioner-facing result: a map of
which strategy to use, and whether availability-aware adaptation is worth its complexity, as a function
of two binary conditions (can optimizer state be preserved, and are gaps short or eclipse-scale?) and
the stationarity of the data.

\begin{table}[t]
\centering
\caption{When is availability-aware adaptation worth it? Synthesis across regimes.}
\label{tab:decision}
\small
\begin{tabular}{llll}
\toprule
\textbf{State preservable?} & \textbf{Gaps} & \textbf{Best strategy} & \textbf{AAT worth it?} \\
\midrule
Yes (full)        & any           & \textsc{Checkpoint-Strong}        & No ($\leq$ neutral) \\
Yes (full)        & any (drift)   & \textsc{Checkpoint-Strong}        & No \\
No (weights-only) & frequent short & \textsc{AAT-drop} (decayed moment) & \textbf{Marginally ($\approx$21\%, noisy)} \\
No (weights-only) & eclipse-scale  & \textsc{Elastic} $\equiv$ \textsc{AAT-drop} & No (decay $\to 0$) \\
\bottomrule
\end{tabular}
\end{table}

The actionable summary: \emph{checkpoint full optimizer state and index your schedule on effective
time}; reach for availability-aware reconstruction only when you genuinely cannot persist optimizer
state (large models, constrained nodes) \emph{and} interruptions are frequent and short relative to the
momentum timescale (thermal throttling, contention, attitude maneuvers), not eclipse-scale outages.

\section{The anticipative direction we leave open}
\label{sec:proactive}

Schedule foreknowledge is real and only one of our methods uses it (\textsc{Proactive}), yet that
prototype did not help (Tables~\ref{tab:stationary},~\ref{tab:costb}): under full preservation there is
nothing to protect, so pre-gap learning-rate ``landing'' is a free perturbation. We do not
claim the foreknowledge angle works; we report only that the naive version does not. The version worth
testing (weight-EMA/SWA~\citep{izmailov2018swa} consolidation across a \emph{known} boundary, and
allocating interruption-sensitive work into predicted-long sessions) requires a regime where the gap is
costly (Cost A or B) for the anticipation to have anything to buy. Establishing that is future work, and
we flag it as an open hypothesis.

\section{Limitations}
\label{sec:limitations}

\textbf{Seed count.} Results are over 3 seeds (CIFAR) / 5 seeds (GPT-2) with error bars; the lone
positive (Cost-A short-gap recovery) is high-variance, so more GPT-2 seeds would tighten it further.
\textbf{Scale.} We test ResNet-18 and a small GPT-2; the
positive Cost-A regime should grow more relevant at larger scale (where optimizer state is genuinely
unpersistable), which strengthens rather than weakens the characterization but should be verified.
\textbf{Drift proxy.} Cost B uses a synthetic per-session input transform, not true on-orbit
distribution shift; a streaming/continual-learning benchmark evaluated on the live distribution is the
natural next test of the BN and momentum pillars. \textbf{Anticipative methods untested.}
Section~\ref{sec:proactive} describes a direction we have not yet validated.

\section{Conclusion}

Specialized availability-aware optimizers are often proposed for training under predictable interruption.
We tested the premise on a physics-grounded benchmark with a strong baseline and found that, when
optimizer state can be preserved, an availability gap is essentially free: competent checkpoint-and-resume
matches uninterrupted training, and reported wins for availability-aware methods trace to a weak
wall-clock baseline. Reactive adaptation helps only in a narrow, identifiable corner: large models with
non-persistable optimizer state under frequent short interruptions, where reconstructing a decayed moment
recovers $\approx$21\% of the state-loss penalty on average (and not reliably) and nothing more. We release OrbitTrace, the strong
baseline protocol, and this decision boundary so the next method is measured against the right bar.

\section*{Acknowledgments}
We thank the RotaStellar team for the CAE planning system from which OrbitTrace is derived. Compute was
provided by RunPod.

\bibliographystyle{plainnat}
\bibliography{references}

\end{document}